\documentclass[reprint]{revtex4-1}
\usepackage[english]{babel}
\usepackage{amsmath,amsthm,amssymb}
\usepackage{amsfonts}
\usepackage{url}
\usepackage{bbm}
\usepackage[colorlinks=true
,urlcolor=blue
,anchorcolor=blue
,citecolor=blue
,filecolor=blue
,linkcolor=blue
,menucolor=blue
,linktocpage=true
,pdfproducer=medialab
,pdfa=true
]{hyperref}

\newcommand{\be}{\beta}

\newcommand{\al}{\alpha}

\newcommand{\si}{{\sigma}}

\newcommand{\ep}{\varepsilon}
\newcommand{\del}{{\partial}}
\newcommand{\HH}{{\mathcal H}}
\newcommand{\mc}[1]{\mathcal{#1}}
\newcommand{\id}{\mathbbm{1}}
\newcommand{\tr}{\mathrm{tr}}
\newcommand{\half}{\frac{1}{2}}
\newcommand{\nn}{\nonumber}
\newtheorem{theorem}{Theorem}[section]
\newtheorem{lemma}[theorem]{Lemma}

\theoremstyle{remark}
\newcommand{\ket}[1]{\left| #1 \right>}
\newcommand{\bra}[1]{\left< #1 \right|} 
\newcommand{\braket}[2]{\left< #1 \vphantom{#2} \right|
	\left. #2 \vphantom{#1} \right>}

\newcommand{\mrm}{\mathrm}
\newcommand{\nodagger}{{\phantom{\dagger}}}
\newcommand{\psipsi}{\ket{\psi}\!\bra{\psi}}

\begin{document}
	\title{Dynamical fidelity susceptibility of decoherence-free subspaces}
	\author{Joris Kattem\"olle, Jasper van Wezel}
	\affiliation{Institute for Theoretical Physics, University of Amsterdam,
		Science Park 904, Amsterdam, The Netherlands}
	\affiliation{QuSoft, CWI,
		Science Park 123, Amsterdam, The Netherlands}
	
	\begin{abstract}
	In idealized models of a quantum register and its environment, quantum information can be stored indefinitely by encoding it into a decoherence-free subspace (DFS). Nevertheless, perturbations to the idealized register-environment coupling will cause decoherence in any realistic setting. Expanding a measure for state preservation, the dynamical fidelity, in powers of the strength of the perturbations, we prove stability to linear order is a generic property of quantum state evolution. The effect of noise perturbation is quantified by a concise expression for the strength of the quadratic, leading order, which we define as the dynamical fidelity susceptibility of DFSs. Under the physical restriction that noise acts on the register $k$-locally, this susceptibility is bounded from above by a polynomial in the system size. These general results are illustrated by two physically relevant examples. Knowledge of the susceptibility can be used to increase coherence times of future quantum computers.
	\end{abstract}	\maketitle
	
	\section{Introduction}
	Quantum computers may be used to simulate general quantum mechanical systems \cite{feynman1982simulating, lloyd1996universal}. Already they drive an enormous effort in the control of quantum systems \cite{Wineland2013} and spur interest in quantum information theory, with connections to high-energy physics \cite{Almheiri2015, Harlow2013, dvali2017} and the quantum-to-classical transition \cite{zurek2003decoherence}.
	
	The biggest roadblock on the way to scalable quantum computation is that of noise and decoherence \cite{nielsen2010quantum, unruh1995maintaining, palma1996quantum}. Quantum error correction offers solutions to this problem \cite{shor1996fault, lidar2013}. In active quantum error correction, errors have to be detected and corrected, whereas in passive error correction, the strategy is to avoid the errors by encoding. The two forms of error correction can be used in conjunction \cite{lidar1999concatenating, kempe2001theory}, and can be described in the same mathematical framework \cite{kribs2005unified}. 
	
	An important player in the passive category is the decoherence-free subspace (DFS) \cite{palma1996quantum, duan1997preserving, zanardi1997error, zanardi1997noiseless,zanardi1998dissipation, lidar1998decoherence, lidar2003decoherence, kempe2001theory}. Although DFSs have been superseded theoretically by more general notions of passive error correction \cite{blume-kohout2010information}, they remain of interest both in theory and in practice \cite{kielpinski2002architecture, blais2004cavity, zhou2004scalable}. In this technique, symmetries of the register-environment coupling are exploited to store quantum information in a register subspace whose reduced time evolution is purely unitary. In contrast to states outside of a DFS, those in it do not suffer from decoherence. Only register-environment models with enough symmetry allow for DFSs.  
	
	In real systems, there are small deviations from the idealized model of the interaction between the quantum register (the `system') and the environment (the `bath'). In particular, these may lead to `super-decoherence' \cite{monz201114-qubit}, where the decoherence time scales adversely with the system size, even for states in a DFS.  
	Quantifying the sensitivity of DFSs to perturbations has lead to the definition of the `dynamical fidelity' \cite{lidar1998decoherence, bacon1999robustness}.
	
	The dynamical fidelity is a measure for the closeness of two states: ($i$) a state, possibly in a DFS, evolving in time under the original model, and ($ii$) the \emph{same} initial state, evolving under the presence of an additional system-bath interaction whose strength is proportional to $\ep$. At the initial time, the dynamical fidelity equals unity, but as time evolves the two states will start to diverge, decreasing the fidelity. The dynamical fidelity can be seen as a generalization of the Loschmidt echo \cite{gorin2006dynamics} to open quantum systems, and is related to a different fidelity in the context of phase transitions \cite{zanardi2007information, you2007fidelity,gu2010fidelity}. Here the fidelity measures the closeness between the ground states of Hamiltonians with different parameter values. 
	
	In an initial qualitative study  \cite{bacon1999robustness}, it was shown the dynamical fidelity can \emph{only} depend linearly on $\ep$ whenever the unperturbed system evolves unitarily on its own in a nontrivial way. This is so for nondegenerate logical states, or whenever the quantum register is used in a quantum computation. Conversely, there is no term linear in $\ep$ whenever the quantum register does not evolve on its own. This led to the conclusion that DFSs are `robust' or `stable' against perturbations when used as quantum memory, but not when used during a quantum computation \cite{lidar1998decoherence, bacon1999robustness, lidar2003decoherence, altepeter2004experimental, gorin2006dynamics, doll2006limitation, schlosshauer2007decoherence, schlosshauer2014quantum}.
	
	Here we prove there never is a linear dependence on $\ep$. In the parlance of the previous work this means DFSs are also stable when used during a quantum computation. However, we show the result even holds for initial states outside of a DFS. In retrospect, the absence of a linear term in the expansion of the dynamical fidelity is a consequence of its definition. This can be considered positive for DFSs, because it shows states in a DFS do not react more strongly to perturbations than regular states. For the fidelity in the context of phase transitions the absence of a linear term is already known \cite{zanardi2007information, you2007fidelity,gu2010fidelity}. 
	
	We go on to introduce the dynamical fidelity susceptibility of DFSs ($
	\chi$), `susceptibility' for short, defined as the strength of the term in the dynamical fidelity proportional to $\ep^2t^2$ \footnote{A quantity called the `dynamical fidelity susceptibility' was also introduced in \cite{mera2018dynamical}. Despite the name, this is something different; it refers to the fidelity between a thermal state (under some Hamiltonian) and the same state after being time evolved by a perturbed Hamiltonian.}. As the first nontrivial term, the susceptibility quantifies the leading order sensitivity to perturbations of states in a DFS. Surprisingly, it does not depend on the unperturbed Hamiltonian, so the leading order behavior of DFSs is as if there were no unperturbed system-bath interaction. Furthermore, it means our result can be used to study the behavior of \emph{any} state under perturbations, outside of the context of DFSs, as long as the unperturbed system-bath interaction vanishes. In that case the DFS of a quantum register is its entire Hilbert space. Even though physically the leading order in time is the most interesting, we later generalize to include all orders in time for completeness. 
	
	For general perturbations, we show the susceptibility is bounded from above by an exponential in the system size, $\chi=O(2^{4n})$, with $n$ the number of qubits. A DFS for which the susceptibility increases exponentially should be considered nonscalable in any practical sense. However, noise typically possesses a form of locality \cite{feynman1982simulating,lloyd1996universal, preskill2013sufficient}, which enforces a more favorable scaling with $n$. Specifically, we consider perturbing Hamiltonians that act on the system $k$-locally, which means they do not contain products of more than $k$ individual qubit operators. In this definition spatial locality is not demanded. The most commonly considered case is $k=1$, which describes completely local noise \cite{nielsen2010quantum}. For general $k$-local perturbing noise, the susceptibility is shown to be bounded from above by a polynomial, $\chi=O(n^{2k})$. This can be compared to the related result on active error correction by Preskill \cite{preskill2013sufficient}, and impacts the scalability of quantum computation \cite{kalai2011quantum, kalai2016quantum} using DFSs.

	To illustrate these results, we compute the susceptibility of a highly nonclassical state, the GHZ state, in two types of DFS. The first protects against pure collective dephasing, the second additionally against collective emission and absorption. We find 
	$\chi=n^2$ and $\chi=4n/3$, respectively. 
Similar scaling laws were found before \cite{palma1996quantum, breuer2002theory, monz201114-qubit} for non-DFS states, and, in fact, DFSs were designed to prevent such scaling laws. Our work shows that these scaling laws are still present in practice.
	
	\section{The dynamical fidelity} Consider a system $S$ in a bath $B$. In the context of quantum computation, $S$ is the collection of qubits, the quantum register, and $B$ is the environment, such as the electromagnetic field. The overall Hilbert space is $\HH=\HH_S\otimes\HH_B$, where $\HH_S$ ($\HH_B$) is the Hilbert space of the $S$ ($B$). We assume $\HH$ to be finite-dimensional unless stated otherwise. (In particular, infinite baths may arise in the context of Lindbladian time evolution.) In general, the Hamiltonian can be written in the form
	\begin{equation}\label{eq:Hamiltonian}
	H_0=H_S\otimes \id + \id\otimes H_B+H_{SB},
	\end{equation}
	where $H_S$ ($H_B$) acts only on $S$ ($B$) and $H_{SB}$ is a system-bath interaction term. In an ongoing quantum computation, $H_S$ includes the generators of the gates.
	
	Now assume that at $t=0$ we have a product state $\rho_{SB,\mrm{init}}= \ket\Psi \bra\Psi$, with $\ket{\Psi}=\ket{\psi}\otimes\ket{\varphi_0}$. For a nontrivial $H_{SB}$ the Hamiltonian     (\ref{eq:Hamiltonian}) will induce entanglement between $S$ and $B$. Tracing out $B$, the pure system state $\psipsi$ at time $t=0$ will be mapped to a mixed system state at $t>0$ by time evolution. We denote this map, or quantum channel, by $\mc A(t)=\mc A$ \cite{nielsen2010quantum}. For every $t\geq 0$ we have a quantum channel. The system state after time $t$ equals $\rho_S(t)\equiv\mc A(\psipsi)=\tr_B(e^{-itH}\ket{\Psi}\bra{\Psi}e^{itH})$ in units where $\hbar=1$. This can be rewritten by introducing the Kraus operators $A_i(t)\equiv \bra{\varphi_i} e^{-itH}\ket{\varphi_0}$, where $\{\ket{\varphi_i}\}$ forms an orthonormal basis for $\HH_B$, with $\ket{\varphi_0}$ the initial bath state. Since $H_0$ acts on both $\HH_S$ and $\HH_B$, and the $\{\ket{\varphi_i}\}$ are bath states, $A_i(t)$ acts nontrivially on $\HH_S$ only. Thus the operator sum representation (OSR) of $\mc A$ is obtained,
	\begin{equation}\label{eq:KrausMap}
	\mc A(\psipsi)=\sum_iA_i(t) \psipsi A_i^\dagger(t).
	\end{equation}
	Because $\mc A$ is trace-preserving, we have the normalization condition $\sum_iA_i^\dagger(t) A_i^\nodagger(t)=\id$. 
	
	In general, $\mc A$ may map pure states to mixed states. A DFS, on the other hand, is defined as a subspace of $D \subset \HH_S$ for which, despite coupling to the bath via $H$,  $\mc A(\psipsi)=e^{-itH_S}\psipsi e^{itH_S}$ for all $\ket \psi \in D$, where $e^{-itH_S}\ket\psi$ has to remain in $D$ \cite{zanardi1997error,lidar1998decoherence}. Thus pure states in a DFS are mapped to pure states in the same DFS by $\mc A$. In terms of the OSR, a necessary and sufficient condition for $\ket \psi\in D$ is $A_j \ket \psi = g_j e^{-i t H_S} \ket \psi $ for all $j$, where $\sum_j |g_j|^2=1$ \cite{lidar2003decoherence, bacon1999robustness}. We do not assume $\ket \psi \in D$ unless stated otherwise. 

	Consider the perturbation $V$ to the Hamiltonian  $H_0$,
	\begin{equation}\label{eq:perturbedHamiltonian}
	H=H_0+\ep V,
	\end{equation}
	where $\ep$ is a real parameter. (The $\ep$-dependence of $H$ is suppressed.) The system state now also depends on $\ep$, and the OSR of the map induced by $H$ is $\rho_S(\ep,t)\equiv\mc A_\ep(\psipsi)=\sum_iA^{\nodagger}_{i}(\ep,t) \psipsi A^{\dagger}_i(\ep,t)$ with $A_i(\ep,t)=\bra{\varphi_i} e^{-itH}\ket{\varphi_0}$. Since the exponential map is analytic, the Kraus operators of the perturbed map may be expanded around $\ep=0$ as \begin{equation}\label{eq:APrime}
	A_i(\ep,t)=A^{(0)}_i(t)+\ep A_i^{(1)}(t)+\ep^2 A^{(2)}_i(t)+O(\ep^3).
	\end{equation}
	Here $A^{(1)}_i(t)=\bra{\varphi_i} -itV-t^2(H_0V+VH_0)/2+O(t^3)\ket{\varphi_0}$.  The explicit form of $A^{(2)}_i(t)$ is of no interest because it will be eliminated. We do not allow qubits to leave the system, so even the perturbed quantum channel needs to be trace-preserving. Thus $\sum_i A^\dagger_i(\ep,t) A^\nodagger_i(\ep,t)=\id$ for all real $\ep$ and $t$. When the expansion (\ref{eq:APrime}) is substituted this imposes
	\begin{align}
	&\sum_i\left(A^{(0)\dagger}_i(t) A^{(1)}_i(t)+A^{(1)\dagger}_i(t) A^{(0)}_i(t)\right)=0,\label{eq:normalizationOne}\\ &\sum_i\left(A^{(0)\dagger}_i(t) A^{(2)}_i(t) + A^{(1)\dagger}_i(t) A^{(1)}_i(t)\right. \nn \\
	&\hspace{9em} \left.  + 
	A^{(2)\dagger}_i(t) A^{(0)}_i(t)\right)=0. \label{eq:normalizationTwo}
	\end{align}
	Conditions involving higher orders of the expansion can be obtained straightforwardly. The above relations are general, since they put constraints on perturbations to general quantum channels, applicable outside the present context. There are no separate conditions that follow from the complete positivity of $\mc A_\ep$; any map that has OSR is automatically completely positive. If one is interested in the effects of a perturbation of the quantum channel rather than a perturbation of the Hamiltonian, expression (\ref{eq:APrime}) is the starting point. 
	
	In general, the fidelity between two states is defined as  $F(\si,\eta)=\left[\tr\sqrt{\sqrt{\si}\eta\sqrt{\si}}\right]^2$ \cite{jozsa1994fidelity}. The effect of a perturbation on the dynamics may be quantified by the dynamical fidelity $F$, that is, the fidelity between the state as obtained after the unperturbed time evolution and the state after the perturbed time evolution, $
	F\equiv F[\rho_S(0,t),\rho_S(\ep,t)]$. Often it is intractable to compute the fidelity because of the square roots. However, if $\rho_{S,\mrm{init}}=\psipsi$, and $\ket \psi$ in in a DFS such that the state remains pure for all $t>0$, the fidelity simplifies to
	\begin{equation}\label{eq:Fsimple}
	F=\bra{\psi(t)} \rho_S(\ep,t)\ket{\psi(t)}
	\end{equation}
	with $\ket{\psi(t)}=U(t)\ket\psi\equiv e^{-itH_S}\ket\psi$.

	\section{Expansion of the dynamical fidelity} The dynamical fidelity $F$ is analytic in $\ep$ at $\ep=0$ because it is a composition of analytical functions of $\ep$. Even though this might seem evident, a more elaborate proof is given in appendix \ref{sec:fidelity}. Now $F$ may be expanded for small $\ep$ if the perturbation is weak,
	\begin{equation}\label{eq:Fexpansion}
	F=1+\ep  F^{(1)} + \ep^2 F^{(2)} +\ldots.
	\end{equation}
	It has previously been shown that $ F^{(1)}=0$ whenever $H_S=0$ and $\ket \psi$ in a DFS, which leaves open the possibility that $ F^{(1)}\neq 0$ when $H_S\neq 0$, even though $\ket \psi$ is in a DFS \cite{bacon1999robustness}. 
	
	However, $ F^{(1)}=0$ in all cases and all times, even without assuming $\ket \psi$ to be in a DFS. This is a direct consequence of the following theorem, together with the fact that $\rho_S(\ep,t)$ is analytic in $\ep$ at $\ep=0$ for all $t$, as is clear from eq. (\ref{eq:APrime}).

    \begin{theorem}
	Let $\{\si(\ep)\}$ be a family of finite-dimensional density matrices that is analytic at $\ep=0$, and let $F[\si(0),\si(\ep)]$ denote the fidelity between $\si(0)$ and $\si(\ep)$. Then $ F^{(1)}\equiv\frac{d}{d\ep} F[\si(0),\si(\ep)]|_{\ep=0}=0$.
	\end{theorem}

	\begin{proof}
    $F[\si(0),\si(\ep)]$ is analytic in $\ep$ at $\ep=0$. Because
    $0\leq F \leq 1$ for any real $\ep$, and $F[\si(0),\si(0)]=1$, it follows that $ F^{(1)}$ must always vanish.
    \end{proof}
    
The theorem also follows from the connection of the fidelity to the Bures metric tensor \cite{hubner1992explicit,paris2009quantum}. We elaborate more on this relation in appendix \ref{sec:bures}. This relation of the theorem to the robustness of DFSs, which we make clear by the elementary considerations above, was not noticed before. It is opposite to previous suggestions that continue to  proliferate in the literature \cite{lidar1998decoherence, bacon1999robustness, lidar2003decoherence, altepeter2004experimental, gorin2006dynamics, doll2006limitation, schlosshauer2007decoherence, schlosshauer2014quantum}.
    
The theorem also applies when time evolution is generated by a Lindbladian $\mc L$. A perturbation of its Lindblad operators \cite{zanardi2015geometry, albert2016geometry} results in $\mc L\rightarrow \mc L(\ep)=\mc L+\ep \mc L' + \ep^2 \mc L''$ for some linear superoperators $\mc L'$ and $\mc L''$. Thus in this setting $\rho_S(\ep,t)=e^{t(\mc L+\ep \mc L' + \ep^2 \mc L'')}\rho_{S,\mrm{init}}$ where  $\rho_{S,\mrm{init}}$ is the initial state at $t=0$. Since $\mc L(\ep)$ and the exponential map are analytic, $\rho_S(\ep,t)$ is analytic in $\ep$ at $\ep=0$ for all $t\geq 0$. Thus $ F^{(1)}=0$ also when time evolution is generated by a Lindbladian, not only for states in a DFS or noiseless subsystems \cite{zanardi2001stabilization}, but for any state. More background is given in appendix \ref{sec:lindblad}.  

We now return to the OSR, and consider $F^{(2)}$. We stress that now we do assume $\ket\psi$ to be in a DFS. Combining (\ref{eq:KrausMap}), (\ref{eq:APrime}), and (\ref{eq:Fsimple}),  we find  $F^{(1)}= \sum_i\bra{\psi} A^{(0)\dagger}_i(t) A^{(1)}_i(t)+A^{(1)\dagger}_i(t) A^{(0)}_i(t) \ket{\psi}$
	and $F^{(2)}=\sum_i\bra\psi A^{(0)\dagger}_i(t) A^{(2)}_i(t)+A^{(2)\dagger}_i(t) A^{(0)}_i(t)\ket{\psi} +|\bra\psi e^{itH_S}A^{(1)}_i(t)\ket\psi|^2$. At this point it seems that $F^{(1)}\neq 0$. By condition (\ref{eq:normalizationOne}) however, it follows that $F^{(1)}=0$. The second condition (\ref{eq:normalizationTwo}) is crucial in obtaining a concise expression for $F^{(2)}$, as it can be used to eliminate $A^{(0)}_i(t)$ and $A^{(2)}_i(t)$. This yields  
	\begin{equation}\label{eq:F2}
	F^{(2)}=-\sum_i \si^2_\psi[U^\dagger(t) A^{(1)}_i(t)],
	\end{equation}
	with $U(t)=e^{-itH_S}$, and $\si^2_\psi[O]\equiv \bra{\psi}O^\dagger O^\nodagger \ket\psi-|\bra\psi O \ket\psi|^2$. Eq. (\ref{eq:F2}) describes the effect of a perturbation to the Kraus operators on the dynamics of states in a DFS. The entire procedure above can be straightforwardly extended to higher orders in $\ep$.
	
	\section{Susceptibility} We now consider the short-time expansion of $F^{(1)}$. The first nonvanishing term is proportional to $t^2$. We define the proportionality constant $\chi$ (with an extra minus sign) as \emph{the dynamical fidelity susceptibility of DFSs}. That is, $\chi=-\frac{1}{4}\frac{\del^2}{\del \ep^2} \frac{\del^2 F}{\del t^2}|_{\ep,t=0}$, so that
	\begin{equation}\label{eq:Fchi}
	F=1-\chi\, \ep^2 t^2 +O(\ep^2t^4).
	\end{equation}
	This is not yet a computation but only a definition. To obtain an expression for $\chi$ involving $H$, note that in general, the perturbing Hamiltonian can be written as $V=\sum_\al S_\al \otimes B_\al$.  We substitute $A^{(1)}_i(t)=\bra{\varphi_i}-itV \ket{\varphi_0}+O(t^2)$ into (\ref{eq:F2}) and collect terms proportional to $\ep^2t^2$. Using the (connected) system correlation function $\mathbf S$ with matrix elements $\mathbf S_{\al\be}=\bra{\psi} S_\al^\dagger S_\be^\nodagger \ket \psi - \bra\psi  S_\al^\dagger \ket\psi \bra\psi  S_\be^\nodagger \ket\psi$ and the bath correlation function $\mathbf B_{\al\be}=\bra{\varphi_0} B_\al^\dagger B_\be^\nodagger \ket{\varphi_0}$, the result can be written as
	\begin{equation}\label{eq:hamiltonianResult}
	\chi=\tr(\mathbf B \mathbf S^T). 
	\end{equation}
Here the trace is not over $\HH_S$ or $\HH_B$ but over the indices of the correlation functions. When $V$ is a simple tensor product, $V=S\otimes B$, this reduces to $\chi=\bra{\varphi}B^2\ket{\varphi}\si^2_\psi [S]$. 

Eq. (\ref{eq:hamiltonianResult}) assumes the initial system state to be in a DFS, but does depend directly on $H_{SB}$. So in particular, it holds for $H_{SB}=0$, in which case the DFS is all of $\HH_{S}$. Thus eq. (\ref{eq:hamiltonianResult}) can be used outside of the context of DFSs to study the effects of perturbative system-bath coupling as long as there is no initial system-bath coupling.
	
 Mathematically, the only restriction on $V$ is its Hermiticity. For $S$ a qubit register with $n$ qubits, any $V$ may be written as $V=\sum_\al c^\al P^\al \otimes B^\al$, with $c^\alpha$ real, and $P^\al$ an element of the Pauli group $\{ \id, \si^x,\si^y,\si^z\}^{\otimes n}$. In this form there are at most $4^{n}=2^{2n}$ linearly independent terms. Under the assumption that adding a qubit does not change how the former qubits couple to the bath, we have that $c^\al$ and $B^\al$ do not depend on $n$. It then follows from eq. (\ref{eq:hamiltonianResult}) that $\chi=O(2^{4 n})$. Now consider the physical restriction that $V$ acts $k$-locally on the system, which means that every $S^\al$ acts on no more than $k$ qubits, with $k$ independent of $n$ \cite{feynman1982simulating, lloyd1996universal}. Then $V$ contains $O(n^k)$ terms. By eq. (\ref{eq:hamiltonianResult}) it thus follows that $\chi=O(n^{2k})$. 
 
	\section{Two examples} 
	Here we calculate $\chi$ explicitly in two examples. Although $\chi$ does not depend on the unperturbed Hamiltonian, we describe possible unperturbed Hamiltonians to give physical context.
	
	For the first example, consider the DFS that is currently used in ion-trap quantum computers \cite{kielpinski2001decoherence, kielpinski2002architecture}. The register-environment model is that of pure collective dephasing \cite{palma1996quantum}, which is the main source of decoherence for unencoded quantum states in this setup \cite{monz201114-qubit}. The coupling term reads
	$H_{SB}= S^z \otimes \sum_k (g_k a_k+ g_k^*a_k^\dagger)$ with $S^z=\sum_{i=1}^n \si^z_i$ the $z$-component of the total spin operator. Here $g_k$ is the register-environment coupling strength, $a_k$ $(a_k^\dagger)$ the annihilation (creation) operator of an electromagnetic mode with wavenumber $k$ and polarization along the $z$-axis, $n$ the number of physical qubits, and $\si^z_i$ the Pauli $z$-operator that only acts on qubit $i$. We assume the qubits to lie along a line, equidistantly separated by distance $d$. There is no spatial dependence of the coupling strengths, and thus all modes that are summed over are assumed to be of long wavelength compared to the total size of the quantum register (i.e. $1/k\gg n d$). Using two physical qubits $(n=2)$, one logical qubit is protected from the decohering influence of $H_{SB}$ by encoding it in the DFS spanned by the logical states $\ket{\bar 0}=\ket{01}$ and $\ket{\bar 1}=\ket{10}$. For $n>2$ even, the qubits are paired, and each pair encodes one logical qubit. The GHZ state is highly nonclassical and known to be highly sensitive to the environment, which is why it is used in quantum metrology \cite{giovannetti2004quantum} and as a probe for the  preservation of coherence \cite{monz201114-qubit}. It can be protected by encoding it as
	 $\ket\psi=(\ket{\bar 0}^{n/2}+\ket{\bar 1}^{n/2})/\sqrt{2}$.
	
	We perturb the model by adding a bosonic mode that couples to the staggered magnetic moment of the system. This corresponds to an  electromagnetic mode with wavelength $\pi/d$ (in units where $c=1$) coupling locally to the individual spin operators, $\ep V=\ep S^{\rm {stag}}\otimes(a_{\pi/d}^\nodagger+a_{\pi/d}^\dagger)$,  where $S^{\rm{stag}}=\sum_{i=1}^n (-1)^i\si^z_i$. We take the state of the perturbing mode to be the vacuum, that is, the state $\ket{\varphi_0}$ such that $a_{\pi/d}\ket{\varphi_0}=0$. (The state of the other modes is irrelevant, see eq. (\ref{eq:hamiltonianResult}).) This state is chosen because it forms a best-case scenario; the thermal bath can at best be at zero temperature. The computation is not more involved when the thermal or number state is assumed. With all definitions in place, we can directly apply formula (\ref{eq:hamiltonianResult}), to find
	\[
	\chi=n^2. 
	\]
	This example saturates the bound on the system size scaling for a completely local noise model.
	
	For our second example, we consider a more general model that, next to dephasing, includes protection against collective absorption and emission of radiation  \cite{zanardi1997noiseless, lidar2003decoherence}. To the best of our knowledge, at the moment this DFS is not used in quantum computers. The coupling term in the unperturbed Hamiltonian reads $H_{SB}=
	\sum_k[g_k S^+a_k+f_k S^-a_k^\dagger+S^z (h_k a_k + h_k^*a_k^\dagger)]$ (with tensor products omitted). Here $S^\pm=\sum_{i=1}^n \si_i^\pm$ excites (relaxes) the system collectively, with $\si_i^\pm=\si_i^x\pm i \si_i^y$ a combination of Pauli operators. Other symbols are defined as before. For four qubits, two logical states that span a DFS that protects against $H_{SB}$ are $\ket{\bar 0}=\ket s \otimes \ket s$, with $\ket s=(\ket{01}-\ket{10})/\sqrt{2}$ a two-qubit singlet state, and $\ket{\bar 1}=(\ket{t_1 t_{-1}}+\ket{t_{-1}t_1}-\ket{t_0t_0})/\sqrt{3}$ a combination of triplet states, with $\ket{ t_{-1}}=\ket{11}$, $\ket{t_0}=(\ket{01}+\ket{10})/\sqrt{2}$ and $\ket{t_{1}}=\ket{00}$. The state we consider here is similar to that in the first example, $\ket \psi=(\ket{\bar 0}^{n/4}+\ket{\bar 1}^{n/4})/\sqrt{2}$. It is in the DFS of $H_{SB}$ for $n\geq 4$ a multiple of 4. It is an encoded GHZ state when the larger DFS is constructed by simple concatenation of single logical qubit DFSs, like in the previous example, but other methods exist \cite{zanardi1997noiseless, lidar2003decoherence}.
	
	As the perturbation, we again consider a staggered field, with $\ep V$ as in the previous example. Also, we assume the perturbing mode to be in the vacuum state. Using (\ref{eq:hamiltonianResult}), a computation shows that
	\[
	\chi=\frac{4}{3}n,
	\]      
	for $n>4$ a multiple of 4. (For $n=4$ the prefactor is different.)

	\section{Generalization to all orders in time}\label{sec:interaction}  Here we generalize our approach to obtain an expression for $F^{(2)}$, as defined in eq. (\ref{eq:Fexpansion}). To do so, we go to the interaction picture, denoted by the superscript $I$. (If there is no superscript denoting the picture the Schr\"odinger picture is always assumed.) In the interaction picture, the initial $S+B$ state at $t=0$, which is equal in any picture, evolves as $\rho^I_{SB}(\ep,t)=U^I(t)\rho_{SB,\mrm{init}}U^{I\dagger}(t)$, with $U^I(t)=e^{itH_0}e^{-itH}=T e^{-i\ep\int_0^t \mrm dt'\, H^I(t')}$. The operator $U^I(t)$ depends also on $\ep$ but this notation is suppressed in $U^I(t)$ and its dependencies. Here $T$ is the time-ordering operator and $H^I(t)$ is the interaction picture Hamiltonian $\ep H^I(t)=e^{itH_0}\ep V e^{-itH_0}$. As before, the Schr\"odinger picture operator $H_0$ contains the system and bath Hamiltonians, and the original coupling $H_{SB}$, against which the system state is protected by the DFS. The perturbed Hamiltonian $H$ contains an extra perturbation $\ep V$ which causes the system state to decohere.
	
	Assuming, as before, that $\rho_{SB,\mrm{init}}=\ket \psi \!\bra \psi \otimes \ket{\varphi_0}\! \bra{\varphi_0}$, and that $\ket \psi$ is in a DFS, we find the dynamical fidelity equals
	\begin{equation}\label{eq:FsimpleInteraction}
	F=\bra \psi \rho^I_S(\ep,t)\ket \psi,
	\end{equation}   
	with $\rho^I_S(\ep,t)\equiv\tr_B[\rho^I_S(\ep,t)]$ the interaction picture system state. Note that, like the expectation value of operators, the fidelity is invariant under change of picture even though states and operators are not. 
	
	The state $\rho^I_S(\ep,t)$ can be expressed as
	\begin{equation*}
	\rho^I_S(\ep,t)=\sum_i A_i^I(\ep,t) \ket{\psi}\!\bra{\psi} A_i^{I\dagger}(\ep,t),
	\end{equation*}
	using the \emph{interaction picture Kraus operators},
	
	\begin{equation*} 
	A_i^I(\ep,t)=\bra{\varphi_i} U^I(t) \ket{\varphi_0}.
	\end{equation*}
	These operators can be expanded using the Dyson series $U^I(t)=\id -i \ep \int_0^t \mrm dt'\, H^I(t')-\frac{\ep^2}{2}T \int_0^{t} \mrm dt' \mrm dt'' \, H^I(t')H(t'')+\ldots$, which yields,
	\begin{equation*}
	A_i^I(\ep,t)=A_i^{I(0)}(t)+ \ep A_i^{I(1)}(t)+\ep^2 A_i^{I(2)}(t)+ \ldots,
	\end{equation*}
	where now $A_i^{I(0)}(t)=A_i^{I(0)}=\braket{\varphi_i}{\varphi_0}$, and
	\begin{equation}\label{eq:A(1)Interaction}
	A_i^{I(1)}(t)=-i\int_0^t \mrm dt'\, \bra{\varphi_i}H^I(t')\ket{\varphi_0}.
	\end{equation}
	A similar expression holds for $A_i^{I(2)}(t)$, but it is of no interest here because it is eliminated by using the normalization conditions (\ref{eq:normalizationOne}), (\ref{eq:normalizationTwo}) in their interaction form, which amounts to putting a superscript $I$ everywhere.
	
	Comparing the expression for $F$ in the Schr\"odinger picture (\ref{eq:Fsimple}) to that in the interaction picture (\ref{eq:FsimpleInteraction}), we see they are essentially equal. The difference is that, in the interaction picture, the extra factor $U(t)$ is absent, and that the state is not $\rho_S(\ep,t)$ but $\rho_S^I(\ep,t)$. Since we have similar expressions for these states in terms of the (interaction picture) Kraus operators, it is straightforward to show that
	\begin{equation}\label{eq:F2interaction}
	F^{(2)}=-\sum_i \si^2_\psi[ A^{I(1)}_i(t)],
	\end{equation}
	with $\si^2_\psi$ defined as before. Of course $F^{(2)}$ itself is invariant under change of picture, it is just the expression that changes form. Also note the absence of $A^{I(2)}_i(t)$ and thus any time ordering. This is due to the normalization conditions (\ref{eq:normalizationOne}), (\ref{eq:normalizationTwo}) in their interaction form. Equation (\ref{eq:F2interaction}) says that the change of fidelity in a DFS, due to an extra system-bath coupling $\ep V$, is proportional to the sum of the auto-correlation functions of the interaction-picture Kraus operators. 
	
	We can gain further insight into the change of fidelity by studying how the $A_i^{I(1)}(t)$ depend on the specific system and bath operators appearing in $V=\sum_\al S_\al \otimes B_\al$. To do so we \emph{define} $S^I_\al(t)$ and $B^I_\al(t)$ by 
	\begin{equation}\label{eq:HI}
	H^I(t)=\sum_\al S^I_\al(t) \otimes B^I_\al(t).
	\end{equation}Given any $H^I(t)$ (or equivalently any $V$), such $S^I_\al(t)$ and $B^I_\al(t)$ can always be found. To see why this is true in principle (depending on the situation much easier methods may exist), note that we can always choose an operator basis for the Hermitian operators on $\HH_{SB}$, so that $H^I(t)=\sum_{ab} h_{ab}(t) \sigma_a \otimes \sigma_b$ with $\{\si_a\},\ \{\si_b\}$ operator bases for $S$ and $B$ respectively, and $\{h_{ab}(t)\}$ real functions of $t$. Then lists of $S_\al^I(t)$ and $B_\al^I(t)$ can be defined so that the sum in (\ref{eq:HI}) has a single index. Note that only in the case that $H_{SB}=0$, from which it follows that $e^{itH_0}=e^{itH_S}e^{itH_B}$, we may choose $S^I_\al(t)=e^{itH_S}S_\al e^{-itH_S}$ and $B^I_\al(t)=e^{itH_B}S_\al e^{-itH_B}$.
	
	Plugging (\ref{eq:HI}) into (\ref{eq:A(1)Interaction}), and the result into (\ref{eq:F2interaction}), we find the generalization of equations (\ref{eq:Fchi}),(\ref{eq:hamiltonianResult}),
	\begin{equation}
	F=1-\ep^2 \int_0^t \mrm dt' \mrm dt''\, \tr[\mathbf B(t',t'')\mathbf S^T(t',t'')]+\ldots,
	\end{equation}
	with correlation functions
	\begin{align*}
	\mathbf B_{\al\be}(t',t'')=& \bra{\varphi_0} B_\al^{I\dagger}(t') B_\be^{I}(t'')\ket{\varphi_0}, \\
	\mathbf S_{\al\be}(t',t'')=&\bra{\psi} S_\al^{I\dagger}(t') S_\be^{I}(t'') \ket \psi\\
	&- \bra\psi S_\al^{I\dagger}(t') \ket\psi \bra\psi  S_\be^{I}(t'') \ket\psi.
	\end{align*}

\section{Conclusion} Using the dynamical fidelity, we quantify the behavior of DFSs under perturbations of the system-bath interaction. The response to perturbations is of second order, the strength of which  we define as the dynamical fidelity susceptibility. It does not depend on the unperturbed system-bath interaction, so to leading order, states in a DFS respond to perturbations as if there were no unperturbed coupling. Our expressions are applicable outside the context of DFSs whenever the perturbation is the only system-bath interaction.

Instead of the robustness or stability of DFSs, we put forward the scaling of the susceptibility with the system size to assess the value of DFSs. For general perturbations, the susceptibility is upper bounded by an exponential in the system size. However, under the restriction of $k$-locality, the upper bound is polynomial. Therefore, DFSs can be considered scalable in theory. It remains to be shown that perturbations can be made sufficiently weak and uncorrelated to allow practical use of DFSs in large-scale quantum computers.

By identifying the `good' DFSs,  the susceptibility is a tool to increase coherence times. Our quantitative results can be generalized to arbitrary system states, and to more general forms of passive error correcting, such as noiseless subsystems. They can also be adjusted to yield the average-case susceptibility or the worst-case susceptibility.

\textbf{Acknowledgments.}
We thank V. Gritsev, W. Buijs\-man and M. Walter for useful discussion, and P. Zanardi for pointing out the connection to the Bures metric. JvW acknowledges support from a VIDI grant financed by the Netherlands Organization for Scientific Research (NWO).


%


\clearpage
\appendix

\section{Analyticity of the Fidelity}\label{sec:fidelity}
Here we prove a lemma concerning the fidelity
\[
F(\rho,\si)=\left[\tr\sqrt{\sqrt{\rho}\si\sqrt{\rho}}\right]^2. 
\]
Note that in the following, we do not assume $\rho$ or $\si$ to be in a DFS.
\begin{lemma}\label{th:analyticity}
	Let $\{\si(\ep)\}$ be a family of finite-dimensional density matrices that is analytic at $\ep=0$. Then the fidelity $F[\si(0),\si(\ep)]$ is analytic at $\ep=0$. 
\end{lemma}
\begin{proof}
	Since $\si(\ep)$ is analytic we may expand it as a power series, $\si(\ep)=\si^{(0)}+\ep \si^{(1)}+\ep^2 \si^{(2)}+\ldots$, where the $\si^{(i)}$ are constant and finite. Suppose $\si^{(0)}$ is given as an $N\times N$ matrix, and let  $\{p_1,\ldots p_m\}$, with $1\leq m \leq N$, be its (not necessarily distinct) nonzero eigenvalues. There exists a basis in which $\si^{(0)}=\mrm{diag}(p_1,\ldots,p_m,0,\ldots,0)$. Naturally, in this basis, $\sqrt{\si^{(0)}}=\mrm{diag}(\sqrt{p_1},\ldots,\sqrt{p_m},0,\ldots,0)$. Note that this is a projector onto the nonzero eigenspace of $\si^{(0)}$. Thus 
	\begin{align*}
	F[\si(0),\si(\ep)]&=\left[\tr\sqrt{(\si^{(0)})^2+\ep\sqrt{\si^{(0)}}\si^{(1)}\sqrt{\si^{(0)}}+\ldots}\right]^2\\
	&\equiv\left[\tr\sqrt{M(\ep)}\right]^2,
	\end{align*}
	where $M(\ep)=M^{(0)}+\ep M^{(1)}+\ldots$, with $M^{(0)}=\mrm{diag}[(p_1)^2,\ldots,(p_m)^2]$. Here we have used the fact that all matrices in the expansion of $\si(\ep)$ are projected onto the zero-eigenspace of $\si^{(0)}$ so that we can reduce the dimension of the matrix under the square root. Thus the $M^{(i)}$ are constant matrices of dimension $m\times m$ (as opposed to $N\times N$), and $M(\ep)$ is Hermitian and analytic. Denote the set of eigenvalues of $M(\ep)$ by $\{a_i(\ep)\}_{i=1}^m$. It follows from theorem 6.1 in Kato (1966) \cite{kato1966perturbation} that the $a_i(\ep)$ are analytic. Since, furthermore, $a_i(0)>0$, there exist a $\delta>0$ such that $a_i(\ep)>0$ for all $\ep$ in the domain $D=(-\delta,\delta)$. In other words, $M(\ep)$ is positive definite and analytic on the domain $D$. Thus the eigenvalues of $\sqrt{M(\ep)}$ are given by $\{\sqrt{a_i(\ep)}\}_{i=1}^m$, which are again all analytic on $D$. Therefore
	\begin{equation}\label{eq:FinEigenvalues}
	F[\si(0),\si(\ep)]=\left[\sum_{i=1}^m \sqrt{a_i(\ep)}\right]^2
	\end{equation}
	is analytic around $\ep=0$. 
\end{proof}

\section{Perturbing a Lindbladian}\label{sec:lindblad} Here we show that there is no term proportional to $\ep$ in the dynamical fidelity $F$ (as defined in the main text) when time evolution is generated by a Lindbladian $\mc L$, without assuming the initial state to be in a DFS.  Lindblad evolution is often used in the context of infinite baths. where excitations are carried away quickly so that they do not back-react on the system. In the Lindblad-setting, or the `semigroup master equation', $\rho_S(t)=e^{t\mc L}\rho_{S,\mrm{init}}$, with 
\begin{align}\label{eq:Lindblad}
\mc L[\rho]&=-i[H,\rho]+\mc D[\rho],\\
\mc D[\rho]&=\sum_{k}\gamma_k\left(L_k\rho L_k^\dagger-\half\{L_k^\dagger L_k,\rho\}\right).\nn
\end{align}
Here $\{\cdot,\cdot\}$ is the anti-commutator, and $\gamma_k>0$. The Lindblad-operators $L_k$ are bounded linear operators on the system's Hilbert space $\mc H_S$, which is of dimension $N$. They do not obey any special relations; the Lindblad equation (\ref{eq:Lindblad}) induces a trace-preserving and completely positive map by design. Usually the $L_k$ are assumed to be orthonormal, but this is not necessary. 

Note that since $\mc L$ is a linear superoperator that acts on density matrices, it may be represented as an $N^2\times N^2$ matrix that acts on the vector $\rho \cong (\rho_{00}, \rho_{01},\ldots, \rho_{NN})^T$, with $\rho_{ij}=\bra{i} \rho \ket {j}$.

We now perturb time evolution by
\begin{align*}
H &\rightarrow H+\ep V,\\
L_k &\rightarrow L_k+\ep L'_k.   
\end{align*}
The result is $\mc L \rightarrow \mc L+\ep \mc L' + \ep^2 \mc L''$ for some finite, constant linear superoperators $\mc L'$ and $\mc L''$. The exponential map of an analytical matrix is analytical. When we see the $\mc L,\,\mc L'$ and $\mc L''$ as matrices, it is thus evident that $\rho_S(\ep,t)=e^{t(\mc L+\ep \mc L' + \ep^2 \mc L'')}\rho_{S,\mrm{init}}$ is analytical  in $\ep$ at $\ep=0$. It is then a direct consequence of the theorem in the main text that $F^{(1)}=0$ also in the Lindblad-setting. 

\section{Alternative derivation of \texorpdfstring{$F^{(1)}=0$}{F(1)=0}}
Here we give an alternative proof to the theorem in the main text in the case that the analytic family under consideration is obtained by a perturbation. Strictly speaking this proof is redundant because a proof was already given in the main text. Nevertheless, the proof here is much more instructive. This is because it shows explicitly how the normalization conditions play a crucial role. Furthermore, it may act as a stepping stone for a more general result; in order to calculate $F^{(2)}$ for general $\ket \psi \in \HH_S$, thus obtaining a generalization valid also for states outside a DFS, essentially the same steps need to be followed as in the following derivation. 

To calculate $F^{(1)}$ explicitly, we adopt the notation from the proof of lemma \ref{th:analyticity} and continue from eq. (\ref{eq:FinEigenvalues}). We consider the time $t\geq 0$ here as fixed, and will drop the notation of $t$. The first order correction to the eigenvalues $a_i(0)$ can be found using standard perturbation theory. Note, however, that in the standard setting one is interested in the corrections to the eigenvalues of a Hamiltonian. Here we are interested in corrections to the eigenvalues of $M^{(0)}$, which is, like a Hamiltonian, a Hermitian linear operator. Note that, in connection to the notation in the proof of  \ref{th:analyticity}, we are now using the explicit states $\si(0)=\rho_S(0,t)=\rho_S(0)$ and $\si(\ep)=\rho_S(\ep,t)=\rho_S(\ep)$. Thus, by standard perturbation theory,
\begin{align*}
a_i(\ep)&=a_i(0)+\ep \bra{i} \sqrt{\rho^{(0)}_S}\rho^{(1)}_S\sqrt{\rho^{(0)}_S}\ket{i}+\ldots\\
&=a_i(0)+\ep\, p_i \bra i \rho^{(1)}_S \ket i+\ldots,
\end{align*}
where 
\[
\rho^{(1)}_S=\sum_j\left(A_j^{(0)} \rho_{S,\mrm{init}}A_j^{(1)\dagger} +A^{(1)}_j \rho_{S,\mrm{init}}A_j^{(0)\dagger} \right),
\]
with $\rho_{S,\mrm{init}}$ the initial system state and, as before (but now using the specific density operator $\rho_S(\ep)$ ), $\rho_S(\ep)=\rho_S^{(0)}+\ep \rho^{(1)}+\ldots$. The system states $\{\ket i\}$ are the nonzero eigenvectors of $\rho_S(0)$ and are thus \emph{all} eigenvectors of $M^{(0)}$.  From the equations above, it follows that 
\begin{align}
F[\rho_S(0),\rho_S(\ep)]&=\left[\sum_{i=1}^m \sqrt{a_i(0)+\ep\, p_i \bra i \rho^{(1)}_S\ket i+\ldots}\right]^2 \nn \\
&=\left[\sum_{i=1}^m \left(p_i+\frac{\ep}{2}\, \bra i \rho^{(1)}_S\ket i+\ldots\right)\right]^2. \label{eq:Fexplicit} 
\end{align}
Again, it seems that $F^{(1)}\neq 0$.
Now either $\rho_S(0)$ is full rank or it is not full rank. Let us first assume it is full rank, that is, assume $m=N$ with $N=\mrm{dim}( \HH_S)$. Then by the normalization conditions in the main text, $\sum_{i=1}^m \bra i \rho^{(1)}_S\ket i=\tr\, \rho^{(1)}_S=0$. Therefore, in this case, $F^{(1)}=0$. Now assume that $ \rho_S(0)=\sum_j A_j^{(0)} \rho_{S,\mrm{init}} A_j^{(0)\dagger}$ is not full rank. We may write $\rho_S(0)=\sum_{k=1}^{m} \mathfrak p_k \ket k \! \bra k$, where $m<N$. We can expand the basis $\{\ket i\}$ to span all of $\HH_S$ (in practice this could be done by a Gram-Schmidt process), and write
\begin{align*}
\sum_{i=1}^m \bra i \rho^{(1)}_S\ket i &= \sum_{i=1}^N \bra i \rho^{(1)}_S\ket i- \sum_{i=m+1}^N \bra i \rho^{(1)}_S\ket i\\
&=- \sum_{i=m+1}^N \bra i \rho^{(1)}_S\ket i\\
&=- \sum_{i=m+1}^N\sum_j (\bra i A_j^{(0)} \rho_S(0) A_j^{(1)\dagger}\ket i+c.c.)\\
&=- \sum_{i=m+1}^N\sum_{j,k} \mathfrak p_k (\bra i A_j^{(0)} \ket k \! \bra k A_j^{(1)\dagger}\ket i+c.c.).
\end{align*}
Here $c.c.$ stands for the complex conjugate of the preceding term. 
For all $m+1\leq i \leq N$, we have by definition that $\bra i \rho^{(0)}_S \ket i=0$. Hence, for these $i$,
\begin{align*}
\bra i \sum_j A_j^{(0)} \rho_S(0) A_j^{(0)\dagger} \ket i&=\sum_{j,k}\mathfrak p_k\bra{i}A_j^{(0)}\ket k \! \bra k A_j^{(0)\dagger}\ket{i}\\
&=\sum_{j,k}\mathfrak p_k\,|\bra i A_j^{(0)} \ket k|^2=0.
\end{align*}
It follows that
\begin{align*}
\bra{i} A_j^{(0)} \ket k = 0
\end{align*}
for all $m+1\leq i < N$ and all $1\leq k \leq m$. Thus, combining the two cases (i.e. $\rho_S(0)$ full rank, $\rho_S(0)$ not full rank), we have
\begin{align*}
\sum_{i=1}^m \bra i \rho^{(1)}_S\ket i=0
\end{align*}
for all $1\leq m \leq N$. Therefore, by eq. (\ref{eq:Fexplicit}), $F^{(1)}=0$ for any $t$ and any perturbation to a quantum channel as defined in the main text, including perturbations obtained by perturbing the overall Hamiltonian.

\section{Scaling of \texorpdfstring{$\chi$}{chi} in the two examples}
In both examples in the main text, $V$ is completely local, but only in the first example the bound on $\chi$ for completely local perturbations is saturated. Even though $\chi$ scales polynomially with $n$ in both examples, the different powers can be an important distinction in practice. The difference can be traced back to the fact that, in the first example, both branches of the superposition that make up the encoded GHZ state are eigenstates of $\si^z_i$. That is, $\si^z_i\ket{\bar 0}^{n/2}=\pm \ket{\bar 0}^{n/2}$ and similarly for $\ket{\bar 1}^{n/2}$. This results in nonzero `inter-block cross terms' such as $\bra{\bar0}^{n/2}\si_i^z\si_j^z\ket{\bar 0}^{n/2}$, for $i,j$ belonging to a different pair of qubits. There are $O(n^2)$ of those terms, and thus $\chi$ scales with $n^2$. In contrast, in the second example, the states $\ket{\bar 0}^{n/4}$ and $\ket{\bar 1}^{n/4}$ are not eigenstates of $\si^z_i$. This leads to vanishing `inter-block cross terms', such as $\bra{\bar0}^{n/4}\si_i^z\si_j^z\ket{\bar 0}^{n/4}$ where $i,j$ belong to different groups of four qubits. When $i=j$, $\si_i^z\si_j^z=\id$. There are $O(n)$ of such terms, and hence $\chi$ scales as $n$. To gain further insight into the susceptibility, one could study whether there are general conditions on $V$ and $\ket\psi$ that can be used to determine the scaling of $\chi$ with $n$ a priori. We leave this for future investigation.  

\section{Relation between \texorpdfstring{$\chi$}{chi} and the Bures metric}\label{sec:bures}
The fidelity can be used to define a distance on the space of $N\times N$ density operators. This is the Bures distance \cite{hubner1992explicit, paris2009quantum}
\[
d_B^2(\rho,\si)=2(1-\sqrt{F(\rho,\si)}).
\]
In the main text we have computed $F=F[\rho_S(0,t),\rho_S(\ep,t)]$, which gives
\begin{align*}
d_B^2[\rho_S(0,t),\rho_S(\ep,t)]&=F^{(2)}(t)\,\ep^2+O(\ep^3)\\
&=[\chi t^2+O(t^3)]\ep^2+O(\ep^3).
\end{align*}
Thus $F^{(2)}(t)$ can be interpreted as (the only entry of) the pullback of the Bures metric tensor on the submanifold $\{\rho(\ep,t)\}_\ep$ at $\ep=0$, 
\[
\mrm d_B^2 (\rho_S(\epsilon,t),\rho(\epsilon+d \epsilon,t))|_{\epsilon=0}=F_2(t)\,\mrm d \epsilon^2=(\chi t^2+\ldots)\mrm d \epsilon^2.
\]
Here we have identified $\ep$ as $d \epsilon$. (We use `$d$' for infinitesimals and `$\mrm d$' for one-forms. Denoting the metric tensor by $\mrm d_B^2 [\rho(\epsilon),\rho(\epsilon+d \epsilon)]$, which is not the square of a one-form, is common notation.) Note that the expression above defines a family of metric tensors, one for every $t$. 

In this geometrical picture, $t$ itself is not a coordinate, like $\epsilon$, because we are never comparing $\rho(0,t)$ and $\rho(\ep,t)$ at different times. The Bures metric tensor being a metric tensor, it may seem obvious that there is no first order dependence of $F$ on $\ep=d\epsilon$. This is ultimately a consequence of the fact that the set of all $N\times N$ density matrices is a Riemannian manifold. However, such an argument requires the machinery of differentiable manifolds. Theorem 1 can be seen as giving a much simpler, more elementary proof that is easy to follow for readers not acquainted with these techniques. To the best of our knowledge, the connection between the pullback of the Bures metric and the `robustness' (i.e. the absence of a term proportional to $\ep$ in $F$) of DFSs, which is the important issue here, was in any case  unnoticed.

\end{document}